\documentclass[submission,copyright,creativecommons]{eptcs}
\providecommand{\event}{ICE 2026} 
 
\usepackage{iftex}
 \usepackage{amsmath}
 \usepackage{amsthm}
 \usepackage{amssymb}
 \usepackage{listings}
 \usepackage{graphicx}

 \newtheorem{definition}{Definition}
  \newtheorem{proposition}{Proposition}
   \newtheorem{theorem}{Theorem}
\theoremstyle{remark}
\newtheorem{remark}{Remark}

\ifpdf
  \usepackage{underscore}         
  \usepackage[T1]{fontenc}        
\else
  \usepackage{breakurl}           
\fi

\title{Modelling Shared-Space Coordination in mCRL2:\\ a Bach-to-mCRL2 Translation Framework}

\author{Corentin Reuther  \qquad Jean-Marie Jacquet
\institute{Nadi Research Institute\\ Faculty of Computer Science, University of Namur\\ Belgium}
\email{\{corentin.reuther,jean-marie.jacquet\}@unamur.be}}

\def\titlerunning{Modelling Shared-Space Coordination in mCRL2}
\def\authorrunning{C. Reuther \& J.-M. Jacquet}
\begin{document}
\maketitle

\begin{abstract}

Although significant research has focused on the theory and implementation of data-based coordination languages, the critical aspect of their automated verification using model-checking techniques remains underexplored, which is essential for ensuring reliability and correctness in distributed systems. 
Existing tools, such as Anemone, provide a solid foundation for reachability-based verification of Bach programs.
While they effectively analyze properties expressed in terms of state attainability, extending support to more expressive temporal specifications—such as liveness properties or invariants over shared space contents—remains an open opportunity. Such extensions are important to capture comprehensive system behaviors, for instance, ensuring that “a request is always matched by a response” or that “no message is silently lost”.
Addressing this limitation, we propose an automated translation from Bach to mCRL2, which explicitly represents the shared space, 
thereby enabling the use of mCRL2's $\mu$-calculus model checker to verify complex properties beyond simple reachability.
Complementing this translation, we introduce a systematic
method to analyze the shared space in mCRL2 concerning data reachability, by directly expressing properties over the shared space's contents in the $\mu$-calculus, thus providing a clearer framework for verification. This approach enables a novel verification process that combines action-based properties with state-based properties over the shared space contents, a largely unexplored area in current coordination-language verification approaches, offering a new dimension of analysis.
\end{abstract}

\section[Introduction and Related Work]{Introduction and Related Work}\label{intro}

In the era of distributed computing, systems confront unprecedented challenges in managing interactions between autonomous components. These systems, where human and technical elements coexist, often exhibit emergent behaviors that can lead to failures if not properly coordinated.
For instance, in a microservices setup, a single faulty interaction can cascade into system-wide service disruptions, as illustrated by a major Amazon Web Services outage in October 2025~\cite{reutersAWS2025}.
In light of this complexity, it appears imperative to use robust coordination models that separate computation and interaction, in order to reason and verify the behavior of systems in a formal way.

The concept of coordination languages originated with Gelernter and Carriero, who introduced Linda in 1985~\cite{gelernter1985generative} as a generative communication model based on a shared tuple space — a logical blackboard where processes deposit and retrieve data without direct coupling~\cite{carriero1989linda,gelernter1992coordination}.


Bach~\cite{BaJa-ICE23, DJL18, BaJa-Anemone-21,JL07}, developed at the University of Namur, is directly inspired by Linda and the algebra of communicating processes (ACP)~\cite{bergstra1985algebra} to propose a modern coordination language. It combines ACP operators (sequential \textit{$;$}, parallel \textit{$||$}, non-deterministic choice \textit{$+$} operators) with Linda-like primitives (\textit{tell}, \textit{get}, \textit{ask}, \textit{nask}) operating on a shared space.

While Bach excels at modeling such interactions, verifying its behavioral properties, such as fairness or absence of losses, has proved challenging with the built-in model checker of Anemone~\cite{BaJa-Anemone-21}, a verification platform developed for Bach. Writing complex temporal properties in this environment turned out to be cumbersome, as the tool primarily focuses on state reachability within the shared space.

This limitation prompted the exploration of mCRL2~\cite{groote2014modeling}, a process-algebra-based tool that employs the $\mu$-calculus for rigorous verification and is equipped with state-of-the-art model checking techniques. While Groote and Mousavi~\cite{groote2014modeling} provide a comprehensive framework for modelling and verifying channel-based communicating systems, Bach processes communicate through a shared space whose contents must themselves be subject to verification — an aspect not addressed by standard mCRL2 methodology.

The question of formally verifying coordination languages based on shared spaces is not new, but remains largely open in its most expressive forms. De Nicola et al.~\cite{de2005formal} tackle a related challenge with STOcKLAIM, a stochastic extension of KLAIM~\cite{de2002klaim}: they model action durations as exponential distributions, translate systems into Continuous Time Markov Chains, and verify quantitative properties such as the probability of reaching a given state within a time bound. This is a genuinely different angle — their concern is \emph{how fast} things happen, while ours is \emph{whether} they happen at all. More fundamentally, their atomic propositions track process locations, not data presence in a shared space. To our best knowledge, directly reasoning over the \emph{contents} of a shared space as part of a temporal verification framework remains unexplored 
— and this is precisely the gap we begin to explore.

This paper develops an automated tool for this translation, available as an artefact in~\cite{repo}. The tool makes the shared space explicit as a parallel process in mCRL2 and encodes Bach primitives as actions, enabling mCRL2 to effectively model Linda-like coordination structures. A key contribution of this paper is extending mCRL2's $\mu$-calculus proofs to incorporate shared content alongside actions, enabling content-aware verification (e.g., querying the presence of specific elements or data values in the shared space) beyond standard action-based checks.
More precisely, the objectives of our work are twofold: on the one hand, to describe the Bach-to-mCRL2 translation, highlighting how observable behavior is preserved and the shared space is explicitly modelled; and, on the other hand, to lay the groundwork for verification of temporal properties, combining action-based and state-based reasoning, preparing for future $\mu$-calculus proofs over shared space contents.

The remainder of this paper is organized in five sections. Section~\ref{sec:Bach} introduces the Bach coordination language, its primitives, composition operators, and trace semantics. Section~\ref{sec:bachTomCRL2} presents the Bach-to-mCRL2 translation, including the explicit encoding of the shared space, the mapping of primitives to actions, and the trace correspondence theorem. Section~\ref{sec:verif} introduces the verification framework, recalling the temporal operators of Anemone and the $\mu$-calculus fragment of mCRL2, and presenting the systematic translation of Bach temporal properties into $\mu$-calculus formulas, including the content-aware approach based on Blackboard instrumentation. Section~\ref{sec:usecase} illustrates the approach with a synchronous Load Balancer case study, exploring how the translation enables the verification of coordination properties combining action-based and content-aware reasoning. Finally, Section~\ref{sec:conclusion} concludes with a discussion of the results and perspectives for future work.

\section[Bach Coordination Language]{Bach Coordination Language: Syntax and Semantics}\label{sec:Bach}
    
    In order to make the paper as self-contained as possible, taking inspiration from \cite{BaJa-ICE23,barkallahsocio,DJL18,BaJa-Anemone-21,JL07}, this section presents the key features of the Bach coordination language, providing the necessary foundation for its translation to mCRL2 and formal verification. Our goal is to define its syntax, operational semantics, and the observable behavior of agents through traces, distinguishing clearly between observable actions (Definition~\ref{def:obs-actions}) and structural constructs (Definition~\ref{def:struct-constructs}), which are formally defined later in this section.
    
    Combining Linda's generative communication model with the Algebra of Communicating Processes (ACP), Bach enables agents to interact exclusively through a \emph{shared space}. Agents use primitives for insertion, removal, and observation, allowing a decoupled approach that is ideal for modeling asynchronous interactions in distributed systems.

    \emph{Scope:}
        This work addresses the core Bach coordination language without extensions such as active data (Multi-Bach), AnimBach, or guarded lists, which address different concerns. The core primitives (\texttt{tell}, \texttt{ask}, \texttt{nask}, \texttt{get}) combined with ACP operators form a sufficient base for distributed coordination verification in mCRL2. The syntax and semantics presented afterwards are adapted from the formal definition in~\cite{barkallahsocio} for the scope and notational conventions of this paper; the underlying language and its semantics are unchanged.

    \subsection{Syntax}

        The syntax of Bach is structured in five conceptual layers, ranging from data representation to process composition. While these layers provide a syntactic view of the language, our presentation is intentionally oriented towards its operational interpretation, as developed in Section~\ref{sec:operationalSemantic}. In particular, rather than focusing solely on grammatical well-formedness, we emphasize the constructs that have a direct impact on execution and observable behavior, which will be essential for the definition of trace semantics and the correspondence theorem (Theorem~\ref{thm:trace}).

        This layered presentation is therefore not purely syntactic: it reflects the structure of the operational semantics by separating data, communication primitives, control structures, and agent-level composition constructs.
        This separation will play a central role in the definition of traces and in the translation into mCRL2, where only primitive communication actions are retained as observable behavior.

        \paragraph{Data layer.}

            The data layer of Bach distinguishes two kinds of identifiers and a notion of structured information term:
            \begin{equation}\label{eq:data-layer}
            \begin{aligned}
            \textit{lid} \ &::=\ \text{string of letters and digits starting with a lowercase letter}\\
            \textit{uid} \ &::=\ \text{string of letters and digits starting with an uppercase letter}\\
            \textit{si} \ &::=\ \textit{lid} \mid 
            \textit{lid}(\textit{si}_1, \ldots, \textit{si}_n)
            \end{aligned}
            \end{equation}

            A \textit{lid} (lowercase identifier) used on its own is called a \textbf{token} in the literature and represents atomic data (e.g., \texttt{request}, \texttt{order}). When a \textit{lid} occurs as the head of a structured term $\textit{lid}(\textit{si}_1, \ldots, \textit{si}_n)$, we call it a \textbf{functor}.
            A \textit{uid} (uppercase identifier) denotes an identifier used as a procedure name (e.g., \texttt{ProducerAgent}, \texttt{LoadBalancer}, see Section~\ref{sec:syntagent}).
            A \textit{si} (\emph{si-term}, short for \emph{structured information term}) is therefore either a token, or a structured term obtained by applying a functor to a sequence of si-terms (e.g., \texttt{request(producer, consumer, id, content)})

            \begin{remark}
                In the remainder of this paper, we frequently refer to the notion of a \emph{shared space}, which is modelled as a multiset of structured information terms (si-terms).
                In the coordination literature, this abstraction is often referred to as a \emph{tuple space}, and si-terms are commonly called \emph{tuples}. These notions are considered equivalent in this work.
            \end{remark}

        \paragraph{Communication primitives.}
        Bach provides four basic communication primitives (\texttt{cprim}) operating on si-terms:

        \begin{equation}
        \textit{cprim} \ ::=\ \texttt{tell}(\textit{si}) \mid 
        \texttt{ask}(\textit{si}) \mid \texttt{nask}(\textit{si}) \mid 
        \texttt{get}(\textit{si})
        \end{equation}

        Their meanings: \texttt{tell(si)} inserts an occurrence of \textit{si} into the shared space (always succeeds); \texttt{ask(si)} checks for the presence of a matching \textit{si} (blocks if absent); \texttt{nask(si)} checks for the absence (blocks if present); and \texttt{get(si)} removes a matching \textit{si} (blocks if absent).

        \paragraph{Control structures.}
        Conditions are built from comparisons over si-terms and combined with Boolean operators:
        \begin{equation}
            \begin{aligned}
            \textit{compop} \ &::=\ <\ \ \mid\ \ \leq\ \ \mid\ \ >\ \ \mid\ \ \geq\ \
            \mid\ \ =\ \ \mid\ \ \neq\\
            \textit{boolop} \ &::=\ \&\ \ \mid\ \ |\  \ \mid\ \ !\\
            \textit{c} \ &::=\ \textit{si}_1\ \textit{compop}\ 
            \textit{si}_2\ \ \mid\ \ c_1\ \textit{boolop}\ c_2
            \end{aligned}
        \end{equation}

        \paragraph{Agents and Procedures.}\label{sec:syntagent}
        Agents are the active entities of Bach. They are built from communication primitives and structural constructs (Definition~\ref{def:struct-constructs}). Communication primitives correspond to observable actions (Definition~\ref{def:obs-actions}) on the shared space, while structural constructs define how agents are composed and controlled.
        The grammar of agents is:

        \begin{equation}
        \label{eq:agents}
        \textit{ag} \  ::=  \begin{array}[t]{l}
        \textit{cprim} \ \
        \mid \ \ \textit{ag}_1\,;\,\textit{ag}_2 \
        \mid \ \ \textit{ag}_1 \parallel \textit{ag}_2 \ \
        \mid \ \ \textit{ag}_1 + \textit{ag}_2 \ \ \mid \\
        \quad \quad c \rightarrow \textit{ag}_1 \diamond \textit{ag}_2 \ \
        \mid \ \ \texttt{sum}\ lid\ \texttt{in}\ D :\ \textit{ag} \ \
        \mid \ \ \textit{pn}(\textit{si}_1, \ldots, \textit{si}_n)
        \end{array}
        \end{equation}

        The composition operators have their usual meaning. Sequential composition $\textit{ag}_1 \,;\, \textit{ag}_2$ runs the first agent and then the second; parallel composition $\parallel$ allows interleaved execution; non-deterministic choice $+$ selects one of the two branches.

        In the guarded choice $\textit{c} \rightarrow \textit{ag}_1 \diamond \textit{ag}_2$, the condition $\textit{c}$ may refer to variables already in scope; if it evaluates to true, $\textit{ag}_1$ executes, otherwise $\textit{ag}_2$. The quantification $\texttt{sum}\ \textit{lid}\ \texttt{in}\ D :\ \textit{ag}$ ranges over all values of the finite domain $D$ (Section~\ref{sec:domain}) assigned to $\textit{lid}$, and behaves as a non-deterministic choice between the corresponding instantiations of $\textit{ag}$ 
        Lastly, the grammar includes procedure calls $\textit{pn}(\textit{si}_1, \ldots, \textit{si}_n)$, which invoke the procedure named $\textit{pn}$ by substituting the si-term arguments $\textit{si}_1, \ldots, \textit{si}_n$ for its parameters and executing its body.
        
        A procedure is introduced by a \emph{procedure equation} of the form:
        \begin{equation}\label{eq:proc-decl}
        \textit{pn}(\textit{si}_1, \ldots, \textit{si}_n) = \textit{ag}
        \end{equation}

        \noindent where \textit{pn} is an uppercase identifier (a \textit{uid}) denoting the procedure name, and the $\textit{si}_i$ are parameters in scope inside the body $\textit{ag}$. The keyword \texttt{proc} is used at the top level of a Bach specification to introduce one or more such procedure equations.

        Procedures are not agents themselves, but definitions that can be invoked by agents through procedure calls. As shown in the grammar of agents (Equation~\ref{eq:agents}), procedure calls may appear inside agents and support recursive definitions.

        \paragraph{Domains.}\label{sec:domain}

        In this paper, domains are treated as finite sets of constants. For instance, in the case study (Section~\ref{sec:usecase}), we define \texttt{MyInteger = \{zero, one, two\}}.
        In the full algebraic specification of Bach~\cite{barkallahsocio}, such a domain is declared using a \texttt{sort} declaration (e.g., \texttt{sort MyInteger = struct zero | one | two}), together with optional \texttt{map} and \texttt{eqn} declarations to define functions over these data types.
        More generally, \texttt{sort} introduces new data types, \texttt{map} defines function symbols over these types, and \texttt{eqn} provides their defining equations.
        
        In this work, we abstract from this algebraic specification and reason directly over finite sets of values. This abstraction is sufficient for our results, since Theorem~\ref{thm:trace} (Trace Correspondence) depends only on observable actions and process composition, and not on the internal algebraic structure of data types.

    \subsection{Operational Semantics}\label{sec:operationalSemantic}
    
        In order to present Bach's operational semantics we use labelled transition systems like other papers in the field~\cite{barkallahsocio}.  We use configurations of the form of $\langle A | \sigma \rangle$, where $A$ is an agent  (i.e., an instance of \textit{ag} as defined in Equation~\ref{eq:agents}) and $\sigma$ is the shared space, modelled as a multiset of si-terms. Transitions represent computation steps, with rules for primitives and operators derived from ACP~\cite{bergstra1985algebra} and Linda-like semantics~\cite{carriero1989linda, gelernter1985generative}.

        The operational semantics of Bach distinguishes two complementary categories of language constructs: those that produce \emph{observable actions} on the shared space, and those that govern the \emph{structural composition} of agents. We make this distinction formal below, together with the notion of a \emph{terminated agent} representing the result of a successful execution.

        \begin{definition}[Observable actions]\label{def:obs-actions}
        The \emph{observable actions} of Bach are the labels produced by the execution of communication primitives. They are denoted by the label $\alpha$ ranging over the four possible labels: $t^+$ for \texttt{tell(t)}, $t^-$ for \texttt{get(t)}, $t^?$ for \texttt{ask(t)}, and $t^{\sim}$ for \texttt{nask(t)}, where $t$ is a si-term.
        \end{definition}

        \begin{definition}[Structural constructs]\label{def:struct-constructs}
        The \emph{structural constructs} of Bach are the operators that compose agents without producing observable actions on their own: sequential composition ($\,;\,$), parallel composition ($\parallel$), non-deterministic choice ($+$), guarded choice ($\rightarrow \diamond$), quantification (\texttt{sum}), and procedure calls. Their behavior is defined inductively from the observable actions of their sub-agents.
        \end{definition}

        \begin{definition}[Terminated agent]\label{def:terminated} 
        We extend the syntax of agents (Equation~\ref{eq:agents}) with a distinguished constant $E$, called the \emph{terminated agent}, representing successful completion. By convention, $E$ admits no outgoing transition: for any shared space $\sigma$ and any label $\alpha$, $\langle E \mid \sigma \rangle \nrightarrow$.
        \end{definition}
        
        The transition rules of Bach are presented in two parts: rules governing the communication primitives (Section~\ref{sec:rules-primitives}), which produce observable actions, and rules governing the structural constructs (Section~\ref{sec:rules-structural}), which combine sub-agents inductively. 
        This separation is central to the trace-based translation into mCRL2 developed in Section~\ref{sec:bachTomCRL2}, where only observable actions are retained in the observable behavior of an agent.

        \subsubsection{Transitions rules for Primitives}\label{sec:rules-primitives}

            Figure~\ref{fig:bach-primitives-rules} specifies the four transition rules governing how Bach primitives interact with the shared space. The rules naturally split into two pairs: the first pair modifies the shared space, while the second pair only observes it.

            The first pair captures the primitives that \emph{modify} the shared space. Rule (T) states that \texttt{tell(t)} always succeeds and adds an occurrence of $t$ to the shared space, producing the labelled transition $t^+$. Conversely, rule (G) describes \texttt{get(t)}, which removes one occurrence of $t$ and produces the label $t^-$. This rule applies only when $t$ is already present in the shared space; otherwise, the agent blocks until another agent inserts it.

            The second pair captures the primitives that \emph{observe} the shared space without modifying it. Rule (A) describes \texttt{ask(t)}: it requires $t$ to be present, observes it without consuming it, and produces the label $t^?$. Its counterpart, rule (N), describes \texttt{nask(t)}: it has the explicit premise $t \notin \sigma$ — that is, it requires the \emph{absence} of $t$ — and produces the label $t^{\sim}$.

            \begin{figure}[!t]
            \centering
                \[
                \begin{array}{ll}
                    (T) \quad 
                    \langle \texttt{tell(t)} \mid \sigma \rangle \xrightarrow{t^+} \langle E \mid \sigma \cup \{t\} \rangle \qquad
                    &
                    (A) \quad
                    \langle \texttt{ask(t)} \mid \sigma \cup \{t\} \rangle \xrightarrow{t^?} \langle E \mid \sigma \cup \{t\} \rangle 
                    \\[12pt]
                    (G) \quad
                    \langle \texttt{get(t)} \mid \sigma \cup \{t\} \rangle \xrightarrow{t^-} \langle E \mid \sigma \rangle
                    &
                    (N) \quad
                    \dfrac{t \notin \sigma}{\langle \texttt{nask(t)} \mid \sigma \rangle \xrightarrow{t^{\sim}} \langle E \mid \sigma \rangle}
                \end{array}
                \]
            \caption{Transition rules for the primitives in Bach~\cite{barkallahsocio}.}
            \label{fig:bach-primitives-rules}
            \end{figure}

        \subsubsection{Transition Rules for Structural Constructs}\label{sec:rules-structural}

            \begin{figure}[!t]
            \centering
            \large
            \[
            \begin{array}{ll}
            (S)\ \ \ \ \quad \frac{\langle A \mid \sigma \rangle \xrightarrow{\alpha} \langle A' \mid \sigma' \rangle}{\langle A;B \mid \sigma \rangle \xrightarrow{\alpha} \langle A';B \mid \sigma' \rangle}\qquad\qquad\qquad
            &
            (P)\ \ \ \ \quad \frac{\langle A \mid \sigma \rangle \xrightarrow{\alpha} \langle A' \mid \sigma' \rangle}
            {\substack{\langle A \mid\mid B \mid \sigma \rangle \xrightarrow{\alpha} \langle A' \mid\mid B \mid \sigma' \rangle \\
            \langle B \mid\mid A \mid \sigma \rangle \xrightarrow{\alpha} \langle B \mid\mid A' \mid \sigma' \rangle}}\\[30pt]
            (C)\ \ \ \ \quad \frac{\langle A \mid \sigma \rangle \xrightarrow{\alpha} \langle A' \mid \sigma' \rangle}
            {\substack{\langle A + B \mid \sigma \rangle \xrightarrow{\alpha} \langle A' \mid \sigma' \rangle\\
            \langle B + A \mid \sigma \rangle \xrightarrow{\alpha} \langle A' \mid \sigma' \rangle}}
            &
            (Co)\ \ \quad \frac{\vDash C,\; \langle A \mid \sigma \rangle \xrightarrow{\alpha} \langle A' \mid \sigma' \rangle}
            {\substack{\langle C \rightarrow A \diamond B \mid \sigma \rangle \xrightarrow{\alpha} \langle A' \mid \sigma' \rangle\\
            \langle ! C \to B \diamond A \mid \sigma \rangle \xrightarrow{\alpha} \langle A' \mid \sigma' \rangle}}\\[30pt]
            (Pc)\ \  \quad \frac{P(\bar{x}) = A,\; \langle A[\bar{x}/\bar{u}] \mid \sigma \rangle \xrightarrow{\alpha} \langle A' \mid \sigma' \rangle}{\langle P(\bar{u}) \mid \sigma \rangle \xrightarrow{\alpha} \langle A' \mid \sigma' \rangle}
            &
            (Sum)\quad \frac{ d \in D,\; \langle A[x/d] \mid \sigma \rangle \xrightarrow{\alpha} \langle A' \mid \sigma' \rangle}{\langle sum\ x\ in\ D:\ A \mid \sigma \rangle \xrightarrow{\alpha} \langle A' \mid \sigma' \rangle}
            \end{array}
            \]
            \caption{Operational semantics rules for the structural constructs of Bach. All rules except (Sum) follow \cite{barkallahsocio}; (Sum) is introduced in this work.}
            \label{fig:bach-composition}
            \end{figure}

            The structural constructs of Bach define how agents are composed and how control flows during execution. Their operational semantics follow an interleaving model inspired by ACP and is given in Figure~\ref{fig:bach-composition}. Rule (S) handles sequential composition $A;B$: the first agent $A$ executes, and the continuation $B$ remains pending. Rule (P) handles parallel composition $A \parallel B$: either side may make a transition, yielding interleaved executions. Rule (C) handles non-deterministic choice $A + B$: either branch may evolve, with the unselected branch discarded. Rule (Co) handles guarded choice $C \rightarrow A \diamond B$: when the condition $C$ holds, $A$ is executed; when it does not, $B$ is. Rule (Pc) handles procedure calls $P(\bar{u})$: the call is unfolded by substituting the arguments into the procedure body $A$. Rule (Sum) handles the quantification operator $\texttt{sum}\ x\ \texttt{in}\ D:\ A$: any value $d \in D$ may be selected, instantiating $A$ accordingly.

            \begin{remark} [On the multi-conclusion formulation]
            The rules (P), (C), and (Co) are presented with two symmetric conclusions to capture the commutativity of $\parallel$, $+$, and $\rightarrow \diamond$ respectively. An alternative formulation would introduce a structural congruence relation $\equiv$ on agents (e.g., $A \parallel B \equiv B \parallel A$) together with single-conclusion rules. 
            \end{remark}

    \subsection{Trace Semantics of Bach}

        \begin{definition}[Trace semantics of Bach]\label{def:bach_sem}
        We define the observational semantics of a Bach agent as the function
        $
        \mathcal{O}_{B} : \textit{Agent} \rightarrow \mathcal{P}(\textit{Trace})
        $,
        where a trace is a finite sequence of observable actions.
        We denote by $\mathcal{P}(\textit{Trace})$ the power set of all such traces, i.e., the set of all possible (non-deterministic) execution behaviors of a given agent. For a Bach agent $A$, we define $\mathcal{O}_{B}(A)$ as follows:
        \[
        \mathcal{O}_{B}(A) =
        \left\{
        \alpha_1 \cdots \alpha_n \;\middle|\;
        \langle A, \emptyset \rangle
        \xrightarrow{\alpha_1}
        \langle A_1, \sigma_1 \rangle
        \xrightarrow{\alpha_2}
        \cdots
        \xrightarrow{\alpha_n}
        \langle A_n, \sigma_n \rangle \nrightarrow
        \right\}.
        \]
        \end{definition}

\section{Bach-to-mCRL2 Translation and Associated Tools}\label{sec:bachTomCRL2}

    This section describes the mapping from Bach to mCRL2 and introduces an automated translation tool. The resulting encoding provides a basis for subsequent analysis in mCRL2, including the verification of temporal properties using the $\mu$-calculus, which is developed later in this paper.

    \subsection{Overview of mCRL2}

        mCRL2~\cite{groote2014modeling} is a process algebra-based formalism and toolset designed for modeling and verifying concurrent systems. 
        Sharing its roots in the Algebra of Communicating Processes (ACP) with Bach, mCRL2 is particularly well-suited for this translation. Its distinguishing feature is the integration of the modal $\mu$-calculus, a highly expressive logic that enables automated verification of complex behavioral properties.
    
        The mCRL2 toolset provides simulation (\texttt{simulator}), state space generation (\texttt{lps2lts}), and property verification (\texttt{pbes2bool}) using $\mu$-calculus proofs. These features allow a more comprehensive analysis of Bach models, enabling developers to explore agent interactions and verify system properties beyond what is possible with the current Anemone tool.

    \subsection{Operational Semantics of the mCRL2 Fragment}

        To establish the correctness of our translation, we first introduce the operational semantics of the relevant fragment of mCRL2, adapted from~\cite{groote2014modeling}.
        Similar to Bach, we represent system behavior using labelled transition systems (LTS), where transitions are labelled by actions $\alpha$, and $E$ denotes the terminated agent (following the convention introduced in Definition~\ref{def:terminated}).
        Figure~\ref{fig:mcrl-composition} presents the rules for the fragment we use.

            \begin{figure}[!t]
            \centering
            \large
            \[
            \begin{array}{ll}
            (T_{mCRL2})  \ \ \quad a \xrightarrow{a} E \qquad \qquad \qquad
            &
            (Al_{mCRL2}) \ \ \quad
            \frac{A \xrightarrow{\alpha} A' \quad \alpha \in S}
                 {\texttt{allow}(S, A) \xrightarrow{\alpha} \texttt{allow}(S, A')}\\[12pt]
            (S_{mCRL2}) \ \ \quad \frac{A \xrightarrow{\alpha} A'}{A.B \xrightarrow{\alpha} A'.B} \qquad \qquad
            &
            (C_{mCRL2}) \ \ \ \ \quad \frac{A \xrightarrow{\alpha} A'}
            {\substack{A + B \xrightarrow{\alpha} A' \\
            B + A \xrightarrow{\alpha} A'}} \\[26pt]
            (P_{mCRL2}^{sync}) \quad  
            \frac{A \xrightarrow{\alpha} A' \qquad B \xrightarrow{\beta} B'}
            {A \parallel B \xrightarrow{\gamma(\alpha,\beta)} A' \parallel B'}
            \qquad \qquad
            &
            (P_{mCRL2}^{L/R}) \quad  
            \frac{A \xrightarrow{\alpha} A' \qquad B \xrightarrow{\beta} B' }
            {\substack{
            A || B \xrightarrow{\alpha} A' || B \\
            A || B \xrightarrow{\beta} A || B'
            }} \\[26pt]
            (Co_{mCRL2}) \ \ \quad \frac{\vDash C, \; A \xrightarrow{\alpha} A'}
            {\substack{C \rightarrow A \diamond B \xrightarrow{\alpha} A' \\
            ! C \to B \diamond A \xrightarrow{\alpha} A'}} \qquad \qquad
            &
            (Pc_{mCRL2}) \quad \frac{ P(\bar{x}) = A \quad A[\bar{x}/\bar{u}] \xrightarrow{\alpha} A'}{ P(\bar{u}) \xrightarrow{\alpha} A'}\\[24pt]
            (Sum_{mCRL2}) \quad\frac{ d \in D \quad A[x/d] \xrightarrow{\alpha} A'}{ sum\ x:D.\ A \xrightarrow{\alpha} A'}
            
            \end{array}
            \]
            \caption{Operational semantics rules for the relevant fragment of mCRL2, adapted from~\cite{groote2014modeling}.}
            \label{fig:mcrl-composition}
            \end{figure}

        Rule $(T_{mCRL2})$ gives the basic action behaviour: an action $a$ executes and reaches the terminated state $E$. Rules $(S_{mCRL2})$, $(C_{mCRL2})$, $(Co_{mCRL2})$, $(Pc_{mCRL2})$, and $(Sum_{mCRL2})$ handle respectively sequential composition, non-deterministic choice, guarded choice $C \rightarrow A \diamond B$, process calls, and quantification operator, with semantics analogous to their Bach counterparts. Two notable differences with respect to Bach are nonetheless worth pointing out, in addition to minor syntactic variations (sequential composition uses the dot operator $A . B$ instead of the semicolon, and quantification follows the syntax $sum\ x:D.\ A$ instead of $sum\ x\ in\ D:\ A$).

        The first difference concerns parallel composition. In addition to independent interleavings — captured by $(P^{L/R}_{mCRL2})$ as in Bach — the rule $(P^{sync}_{mCRL2})$ allows two parallel components to synchronise their actions according to a communication function $\gamma$, which maps pairs of synchronisable actions to a combined action. The function $\gamma$ is defined explicitly via the \texttt{comm} construct of mCRL2 (see Figure~\ref{fig:blackboard}); when $\gamma(\alpha, \beta)$ is undefined, no synchronisation is possible and the rule does not apply.
        This synchronisation mechanism is specific to mCRL2 and is central to our translation, as we will use it to model the interaction between agents and the Blackboard (Section~\ref{sec:blackboardImplementation}).

        The second difference is the \texttt{allow} operator, governed by rule $(Al_{mCRL2})$. It does not merely hide actions outside the set $S$ from external observation: actions not in $S$ are entirely \emph{prevented from occurring}. In other words, $\texttt{allow}(S, A)$ behaves like $A$ except that any transition labelled with an action $\alpha \notin S$ is blocked at the source. This operator plays a key role in defining the observable semantics of the translated system.

    \subsection{Trace Semantics of mCRL2}
    
        \begin{definition}[Trace semantics of mCRL2]
        We define the observational semantics of an mCRL2 agent as the function
        $
        \mathcal{O}_{M} : \textit{Agent} \rightarrow \mathcal{P}(\textit{Trace})
        $,
        where a trace is a finite sequence of observable actions. For an mCRL2 agent $A$,  we define $\mathcal{O}_{M}(A)$ as follows:
        \[
        \mathcal{O}_{M}(A) =
        \left\{
        \alpha_1 \cdots \alpha_n \;\middle|\;
        A
        \xrightarrow{\alpha_1}
        A_1
        \xrightarrow{\alpha_2}
        \cdots
        \xrightarrow{\alpha_n}
        A_n
         \nrightarrow
        \right\}.
        \]
        \end{definition}

    \subsection{Implementation of the Shared Space in mCRL2}\label{sec:blackboardImplementation}

        In a Bach system, a set of agents $A_1, A_2, \ldots, A_n$ execute in parallel and interact exclusively through the shared space. They can be summarized as a single composed agent $A = A_1 \parallel A_2 \parallel \cdots \parallel A_n$. To faithfully translate this into mCRL2, we implement the shared space as an explicit parallel process, called the Blackboard. This design, illustrated in Figure~\ref{fig:blackboard}, separates agent behavior from state management, preserving the decoupled coordination semantics of Bach.
        The initial system is defined as the parallel composition of the Blackboard process and the translated agent $A$, which corresponds to the mCRL2 encoding of the composed Bach agent $A_1 \parallel \cdots \parallel A_n$. The Blackboard is initialized with an empty multiset, denoted \texttt{\{:\}}, which in mCRL2 represents the empty shared space $\emptyset$.

    \begin{figure}[!t]
        \centering
        \begin{lstlisting}[basicstyle=\footnotesize\ttfamily]
proc Blackboard(myBB: BB) = 
        sum x: Elem.
            tb(x) . Blackboard(myBB + {x:1})
        + sum x: Elem.
            (count(x, myBB) > 0) -> gb(x) . Blackboard(myBB - {x:1})
        + sum x: Elem.
            (count(x, myBB) > 0) -> ab(x) . Blackboard(myBB)
        + sum x: Elem.
            (count(x, myBB) == 0) -> nab(x) . Blackboard(myBB);
    
init allow({Told, Got, Asked, Nasked},
    comm({
        tb | tell -> Told, gb | get  -> Got,
        ab | ask  -> Asked,  nab | nask -> Nasked
        },
        Blackboard({:}) || A
    ));
        \end{lstlisting}
        \caption{Explicit implementation of a shared space in mCRL2}
        \label{fig:blackboard}
    \end{figure}

        Each primitive in Bach (\texttt{tell}, \texttt{get}, \texttt{ask}, \texttt{nask}) is translated into an action emitted by the agent process. These actions synchronize with corresponding internal Blackboard actions, producing observable events:
    
        \begin{itemize}
            \item \texttt{tell(t)}: synchronizes with Blackboard action \texttt{tb(t)} and performs observable action \texttt{Told(t)}
            \item \texttt{get(t)}: synchronizes with \texttt{gb(t)} and performs observable action \texttt{Got(t)}
            \item \texttt{ask(t)}: synchronizes with \texttt{ab(t)} and performs observable action \texttt{Asked(t)}
            \item \texttt{nask(t)}: synchronizes with \texttt{nab(t)} and performs observable action \texttt{Nasked(t)}
        \end{itemize}

        \noindent This interaction is modelled in mCRL2 through a synchronization mechanism between the agent and the Blackboard, followed by a restriction of the resulting action set to the intended observable behavior (via the \textit{allow} structure). The Blackboard maintains the multiset of si-terms internally and updates it according to agent actions, while the synchronized actions serve as the observable interface for verification.

        \subsection{Trace Semantics Revisited}
            To support the inductive proof of the trace correspondence theorem, we introduce a generalized version of the Bach trace semantics that takes the initial shared space as an explicit parameter.
            
            \begin{definition}[Generalized Trace Semantics of Bach, $\mathcal{O}'$]
                Let $A$ be a Bach agent and $\sigma$ a shared space. We define:
                \[
                \mathcal{O}'(A)(\sigma) =
                \left\{
                \alpha_1 \cdots \alpha_n \;\middle|\;
                \langle A \mid \sigma \rangle
                \xrightarrow{\alpha_1}
                \langle A_1 \mid \sigma_1 \rangle
                \xrightarrow{\alpha_2}
                \cdots
                \xrightarrow{\alpha_n}
                \langle A_n \mid \sigma_n \rangle \nrightarrow
                \right\}
                \]
                where the condition $\langle A_n \mid \sigma_n \rangle \nrightarrow$  includes the case $n = 0$, i.e., when $\langle A \mid \sigma \rangle \nrightarrow$ itself, yielding the empty trace $\epsilon$.

                Note that the original trace semantics $\mathcal{O}_B(A)$ defined in Definition~\ref{def:bach_sem} 
                corresponds to the special case $\mathcal{O}'(A)(\emptyset)$.
            \end{definition}
            
            \begin{definition}[Recursive Trace Semantics of Bach, $\mathcal{O}''$]
                Let $A$ be a Bach agent and $\sigma$ a shared space. We define $\mathcal{O}''(A)(\sigma)$ recursively as:
                \[
                \mathcal{O}''(A)(\sigma) = 
                \left\{ \epsilon   \;\middle|\;  \langle A \mid \sigma \rangle \nrightarrow \right\}
                \cup
                \left\{
                \alpha_1 \cdot h \;\middle|\;
                \langle A \mid \sigma \rangle
                \xrightarrow{\alpha_1}
                \langle A_1 \mid \sigma_1 \rangle
                \text{ and }
                h \in \mathcal{O}''(A_1)(\sigma_1)
                \right\}
                \]
                where $\epsilon$ denotes the empty trace. The term $\{ \epsilon \}$ serves as the base case of the recursion: it is trivially included, corresponding to the case $\langle A \mid \sigma \rangle \nrightarrow$, just as in $\mathcal{O}'$ where the empty trace arises implicitly when $n = 0$.
            \end{definition}
            
            \begin{proposition}[Equivalence of Trace Semantics] \label{prop:equiv}
                For any Bach agent $A$ and shared space $\sigma$:
                \[
                \mathcal{O}'(A)(\sigma) = \mathcal{O}''(A)(\sigma)
                \]
            \end{proposition}

            \begin{proof}
            Simple verification.
            \end{proof}

    \subsection{Translation Functions Between Trace Semantics}

        To formally relate the traces of a Bach agent to those of its mCRL2 translation, we introduce two functions that map observable actions between the two formalisms.
        
        \begin{definition}[Bach-to-mCRL2 action translation, $b2m$]
            The function $b2m$ maps Bach observable actions to mCRL2 observable actions:
            
            \[
            \begin{array}{lcl@{\qquad\qquad}lcl}
                b2m(t^+) &=& \texttt{Told}(t) & b2m(t^-) &=& \texttt{Got}(t)\\
                b2m(t^?) &=& \texttt{Asked}(t) & b2m(t^{\sim}) &=& \texttt{Nasked}(t) 
            \end{array}
            \]
            It extends to traces as follows:
            $
                b2m(\alpha_1 \cdots \alpha_n) = b2m(\alpha_1) \cdots b2m(\alpha_n)
            $.
        \end{definition}
        
        \begin{definition}[mCRL2-to-Bach action translation, $m2b$]
            The function $m2b$ maps mCRL2 observable actions to Bach observable actions:

            \[
            \begin{array}{lcl@{\qquad\qquad}lcl}
                m2b(\texttt{Told}(t)) &=& t^+ & m2b(\texttt{Got}(t)) &=& t^-\\
                m2b(\texttt{Asked}(t)) &=&  t^?  & m2b(\texttt{Nasked}(t)) &=& t^{\sim}
            \end{array}
            \]
            It extends to traces as follows:
            $
                m2b(\beta_1 \cdots \beta_n) = m2b(\beta_1) \cdots m2b(\beta_n)
            $.
        \end{definition}
        
        \noindent For conciseness, we introduce the following notation for the translated system:
        
        \begin{definition}[Translated system, $\mathit{Act_{tr}}$]\label{def:acttr}
            Let $A$ be a Bach agent. We denote by $\mathit{Act_{tr}}(A)$ the mCRL2 process:
           \begin{eqnarray*}
            \mathit{Act_{tr}}(A) & = & \texttt{allow}(\{\texttt{Told}, \texttt{Got}, \texttt{Asked}, \texttt{Nasked}\}, \\ & & \hspace*{1ex}  
                \texttt{comm}(\{tb \mid tell \rightarrow \texttt{Told},\;
                    gb \mid get \rightarrow \texttt{Got},\;
                    ab \mid ask \rightarrow \texttt{Asked},\;
                    nab \mid nask \rightarrow \texttt{Nasked}\},\; \\ & & \hspace*{2ex}  
                \textit{BB}({:}) \parallel tr(A)))
            \end{eqnarray*}
            where $tr(A)$ is the mCRL2 translation of the Bach agent $A$, and $\textit{BB}(\{:\})$ is the Blackboard initialized with the empty shared space.
        \end{definition}
        The function $tr$ is defined by structural induction on the grammar of agents (Equation~\ref{eq:agents}), mapping each Bach primitive to its corresponding mCRL2 action and each composition operator to its mCRL2 counterpart.

        \begin{theorem}[Trace Correspondence]\label{thm:trace}
            Let $A$ be a Bach agent and $\sigma$ a shared space. Then:
            \begin{enumerate}
                \item For any trace $\alpha_1 \cdots \alpha_n \in \mathcal{O}'(A)(\sigma)$, there exists a trace $\beta_1 \cdots \beta_n \in \mathcal{O}_M(\mathit{Act_{tr}}(A))$ such that:
                \[
                    \alpha_1 \cdots \alpha_n = m2b(\beta_1 \cdots \beta_n)
                \]
                \item For any trace $\beta_1 \cdots \beta_n \in \mathcal{O}_M(\mathit{Act_{tr}}(A))$, there exists a trace $\alpha_1 \cdots \alpha_n \in \mathcal{O}'(A)(\sigma)$ such that:
                \[
                    \beta_1 \cdots \beta_n = b2m(\alpha_1 \cdots \alpha_n)
                \]
            \end{enumerate}
        \end{theorem}
        
        \begin{proof}

             Both directions are proved by induction on the length of the trace, using the recursive characterization $\mathcal{O}''$ established in Proposition~\ref{prop:equiv}, which allows us to decompose any trace into its first action and a continuation.

            \medskip
            \textsc{Direction 1 ($\mathcal{O}' \rightarrow \mathcal{O}_M$).}
            We prove that for any trace $\alpha_1 \cdots \alpha_n \in \mathcal{O}'(A)(\sigma)$, the trace $b2m(\alpha_1 \cdots \alpha_n)$ belongs to $\mathcal{O}_M(\mathit{Act_{tr}}(A))$.
    
            \textit{Base case.} If $\alpha_1 \cdots \alpha_n = \epsilon$, then $\langle A \mid \sigma \rangle \nrightarrow$, and correspondingly $\mathit{Act_{tr}}(A)$ cannot perform any observable action, so $\epsilon \in \mathcal{O}_M(\mathit{Act_{tr}}(A))$.

            \textit{Inductive step.} Assume the property holds for all traces of length $n$. Consider a trace $\alpha_1 \cdots \alpha_{n+1} \in \mathcal{O}''(A)(\sigma)$. By definition of $\mathcal{O}''$, there 
            exists a transition:
            \[
                \langle A \mid \sigma \rangle \xrightarrow{\alpha_1} 
                \langle A_1 \mid \sigma_1 \rangle
            \]
            with $\alpha_2 \cdots \alpha_{n+1} \in \mathcal{O}''(A_1)(\sigma_1)$. The translation of this primitive or structural step produces a corresponding transition in $\mathit{Act_{tr}}(A)$ with label $b2m(\alpha_1)$. By the induction hypothesis applied to $A_1$, the suffix $\alpha_2 \cdots \alpha_{n+1}$ maps to a valid mCRL2 trace. Hence $b2m(\alpha_1 \cdots \alpha_{n+1}) \in \mathcal{O}_M(\mathit{Act_{tr}}(A))$.

            \medskip
            \textsc{Direction 2 ($\mathcal{O}_M \rightarrow \mathcal{O}'$).}
            Symmetric to Direction 1, replacing $b2m$ with $m2b$ and $\mathcal{O}''$ with $\mathcal{O}_M$.

    \end{proof}

\section{Verification of Temporal Properties}\label{sec:verif}

    This section presents the verification framework that bridges Bach's shared-space-oriented temporal logic and mCRL2's $\mu$-calculus. We first introduce the temporal operators supported by Anemone, then introduce the relevant fragment of the $\mu$-calculus used in mCRL2, and finally present the systematic translation of Bach temporal properties into $\mu$-calculus formulas.

  \subsection{Temporal Logic in Bach and Anemone}

         The Anemone model checker~\cite{BaJa-Anemone-21} supports a set of temporal operators oriented towards the shared space contents, making them natural for Bach programs. The basic building block is the predicate $\#x = n$, which checks whether the number of occurrences of si-term $x$ in the shared space equals $n$. This predicate can be generalized using comparison operators ($=$, $\neq$, $<$, $\leq$, $>$, $\geq$) on si-term occurrences, and further combined using logical connectives (conjunction (\texttt{\&}), disjunction (\texttt{|}), and negation (\texttt{!})) to form composed predicates such as \texttt{\#a=1 \& \#b>=2}.

        To express richer properties, three core temporal operators are provided:

        \begin{itemize}
            \item \textbf{Reach $\phi$}: verifies that there exists an execution path reaching a state where $\phi$ holds in the shared space. For example, \texttt{Reach \#a=1} checks whether si-term $a$ can eventually be present.
            \item \textbf{Next $\phi$}: verifies that $\phi$ holds in the immediate next state. For example, \texttt{Next \#a=1} checks whether $a$ appears in the shared space after the very next action.
            \item \textbf{$\phi_1$ Until $\phi_2$}: verifies that $\phi_1$ holds continuously until $\phi_2$ becomes true, and that $\phi_2$ eventually holds. For example, \texttt{(\#a=0) Until (\#b=1)} asserts that $a$ is absent until $b$ appears.
        \end{itemize}

        \noindent While this grammar is naturally si-term-oriented, it suffers from significant limitations. The \texttt{Reach} operator only proves the existence of a path where a property holds, not its universal establishment across all executions. Furthermore, multiple consecutive \texttt{Next} operators and combined \texttt{Reach} and \texttt{Until} are not permitted. These restrictions make the grammar incapable of verifying complex properties such as liveness or universal invariants, motivating the translation to mCRL2.

    \subsection{The $\mu$-Calculus Fragment of mCRL2}

        The mCRL2 toolset uses the modal $\mu$-calculus as its property specification language. Starting from a Linear Process Specification (LPS), the model checker translates a $\mu$-calculus formula into a Parameterized Boolean Equation System (PBES), whose solution determines whether the property holds across all reachable states~\cite{groote2009analysis}. Whenever \texttt{pbes2bool} returns an answer, it corresponds to a rigorous proof, although termination is not guaranteed for infinite or unbounded state spaces.

        The basic modal operators are:
        \begin{itemize}
            \item $[a]\phi$: in every execution of action $a$, property $\phi$ must hold (\emph{universal modality}).
            \item $\langle a \rangle \phi$: there exists at least one execution of action $a$ after which $\phi$ holds (\emph{existential modality}).
            \item $a^*$ / $a^+$: zero or more / one or more repetitions of $a$.
            \item $a_1\,.\,a_2$: sequential composition of action formulae inside a modality, so that $\langle a_1\,.\,a_2 \rangle \phi$ is equivalent to $\langle a_1 \rangle \langle a_2 \rangle \phi$.
        \end{itemize}

        \noindent These operators combine with logical connectives (conjunction \texttt{\&\&}, disjunction \texttt{||}, negation \texttt{!}) and quantifiers (\texttt{forall}, \texttt{exists}). The $\mu$-calculus further provides two fixpoint operators, optionally parameterised: $mu X(p : T = v).\ \phi$ and $nu X(p : T = v).\ \phi$ denote respectively the \emph{least} and \emph{greatest} fixpoint, where the parameter $p$ of type $T$ initialised to $v$ carries state across recursive unfoldings. The least fixpoint expresses liveness (\emph{eventually}), the greatest fixpoint expresses safety (\emph{invariantly}). Boolean expressions are lifted into modal formulae via $\texttt{val}(\phi)$, which holds iff $\phi$ evaluates to true.

        For instance, $nu X.\ ([\mathit{true}]X\ \&\&\ \langle \mathit{true}\rangle \mathit{true})$ expresses deadlock freedom — from any reachable state, at least one action is always enabled — while $mu X.\ ([\mathit{true}]X\ \lvert\rvert\ \langle \mathit{true}^*\ .\ \texttt{Told}(t) \rangle \mathit{true})$ expresses the eventual insertion of a si-term $t$ into the shared space.

    \subsection{Mapping Temporal Properties}

        The Bach temporal operators \texttt{Reach}, \texttt{Next}, and \texttt{Until} translate systematically into $\mu$-calculus formulas, as summarized in Figure~\ref{fig:translations}. These translations apply to any formula $\phi$, whether it concerns coordination actions or shared space contents.
        
        \begin{figure}[h]
        \centering
        \[
        \begin{array}{lcl}
            \texttt{Reach}\ \phi 
            & \longrightarrow 
            & \langle \mathit{true}^* \rangle\, \phi \\[8pt]
            \texttt{Next}\ \phi 
            & \longrightarrow 
            & \langle \mathit{true} \rangle\, \phi \\[8pt]
            \phi_1\ \texttt{Until}\ \phi_2 
            & \longrightarrow 
            & mu X.\ \big(\, \phi_2\ \lvert  \rvert\ (\phi_1 \ \&\&\ \langle 
            \mathit{true} \rangle X)\,\big)
        \end{array}
        \]
        \caption{Translation of Bach temporal operators into $\mu$-calculus.}
        \label{fig:translations}
        \end{figure}

        \texttt{Reach}\ $\phi$ holds if there exists a reachable state 
        satisfying $\phi$. \texttt{Next}\ $\phi$ holds if $\phi$ is satisfied 
        in some immediate successor state. $\phi_1\ \texttt{Until}\ \phi_2$ 
        holds if either $\phi_2$ already holds, or $\phi_1$ holds and the 
        property continues in the next state — the least fixpoint ensuring 
        that $\phi_2$ is eventually reached.

        \paragraph{Content-Aware Verification.}
            A key contribution of this work is the extension of the above translations to properties that directly concern the \emph{contents of the shared space}.  This multiset verification dimension appears to be absent from the mCRL2 literature at the time of the research.
             To support such properties, we extend the Blackboard process with  additional observable actions, triggered whenever specific conditions on the shared space contents are satisfied. For instance, the predicates $\#a = 1$ and $\#b \geq 2$ are encoded as follows:

            \begin{lstlisting}[basicstyle=\footnotesize\ttfamily]
 proc Blackboard(myBB: BB) = 
      ...
    + sum x: Elem. (count(a, myBB) == 1) -> BB_a_equal_1 . Blackboard(myBB)
    + sum x: Elem. (count(b, myBB) >= 2) -> BB_b_sup_2 . Blackboard(myBB);
            \end{lstlisting}

            \noindent When $\phi$ is a predicate over the shared space contents, it is encoded as an observable action prefixed with \texttt{BB\_}, emitted by the Blackboard whenever the condition holds. The corresponding state formula in the $\mu$-calculus is therefore $\langle \texttt{BB\_}\phi \rangle\ \mathit{true}$, which asserts that the action $\texttt{BB\_}\phi$ is available in the current state — and thus, by construction, that $\phi$ holds. Applying this convention to the general translations of Figure~\ref{fig:translations} yields:
            \[
            \begin{array}{lcl}
                \texttt{Reach}\ \#a = 1 
                & \longrightarrow 
                & \langle \mathit{true}^* \rangle\ \langle\texttt{BB\_a\_equal\_1}\rangle\ \mathit{true}\\[8pt]
                \texttt{Next}\ \#a = 1 
                & \longrightarrow 
                & \langle \mathit{true} \rangle\ \langle\texttt{BB\_a\_equal\_1}\rangle\ \mathit{true} \\[8pt]
                (\#a=1)\ \texttt{Until}\ (\#b \geq 2)
                & \longrightarrow 
                & mu X.\ \big( \langle\texttt{BB\_b\_sup\_2} \rangle \mathit{true}\  \lvert  \rvert\ 
                (\langle \texttt{BB\_a\_equal\_1} \rangle \mathit{true}\ \&\& \
                \langle \mathit{true} \rangle X)\big)
            \end{array}
            \]

\section{Case Study: Synchronous Load Balancer in a Microservice Architecture}\label{sec:usecase}

    \subsection{Motivation}

        Microservice architectures are, by design, distributed systems in which different services may depend on one another. A common coordination challenge in such architectures is the fair distribution of incoming requests across multiple service replicas, in order to avoid overloading a single instance while others remain idle. The Load Balancer pattern addresses this challenge by acting as an intermediary that routes requests to available consumers according to a scheduling strategy.

        In this case study, we model a synchronous Load Balancer using a round-robin routing strategy. Each request undergoes a complete processing cycle: the producer initiates the communication, the load balancer forwards the request to a consumer, the consumer responds, the load balancer forwards the response back, and the producer finalizes  the cycle. This synchronous request-response pattern is representative of typical production scenarios, for instance in payment service architectures where a response is required to confirm transaction success before initiating subsequent operations such as order confirmation or shipping.
        We choose the Load Balancer for its conciseness and expressiveness: its round-robin routing and synchronous request-response cycle provide enough behavioral richness to illustrate content-aware verification, while remaining fully presentable within this paper. More complex patterns (Circuit Breaker, Message Broker) are provided as exploratory examples in the associated repository~\cite{repo}.

    \subsection{Bach Model}

\begin{figure}[!t]
        \centering
        \begin{lstlisting}[basicstyle=\small\ttfamily]
proc
    
ProducerAgent(producerId: Service, reqId: MyInteger, 
                                                maxReqId: MyInteger) =
        tell(request(producerId, loadBalancer, reqId));
        sum consumerId in Service:
            get(response(loadBalancer, producerId, consumerId, reqId));
            tell(success(consumerId, producerId, reqId));
            (!(toInt(reqId) = toInt(maxReqId))) -> 
                ProducerAgent(producerId, successor(reqId), maxReqId).


LoadBalancer(roundRobinIndex: MyInteger) =
        sum  producerId in Service, consumerId in Service, 
                                                   reqId in MyInteger:
            (get(request(producerId, loadBalancer, reqId));
            
            (
                (mod2(roundRobinIndex) = zero) ->
                tell(request(loadBalancer, czero, producerId, reqId));
                LoadBalancer(successor(roundRobinIndex))
            <>
                (mod2(roundRobinIndex)= one) ->
                tell(request(loadBalancer, cone, producerId, reqId));
                LoadBalancer(successor(roundRobinIndex))
            ))
            +
            get(response(consumerId, loadBalancer, producerId, reqId));
            tell(response(loadBalancer, producerId, consumerId, reqId));
            LoadBalancer(roundRobinIndex).

ConsumerAgent(consumerId: Service) =
    sum producerId in Service, reqId in MyInteger :
      get(request(loadBalancer, consumerId, producerId, reqId));
      tell(response(consumerId, loadBalancer, producerId, reqId));
      ConsumerAgent(consumerId).
        \end{lstlisting}
        \caption{Bach procedures of the synchronous Load Balancer case study. The full specification including data definitions is available in~\cite{repo}.}
        \label{fig:lb-bach}
        \end{figure}

\begin{sloppypar}
  The system is initialised as the parallel composition of four agents: $\texttt{ProducerAgent(pzero, zero, two)}$, $\texttt{LoadBalancer(zero)}$, $\texttt{ConsumerAgent(czero)}$, and $\texttt{ConsumerAgent(cone)}$. The producer emits three requests (indexed \texttt{zero} to \texttt{two}) and the load balancer starts at an even round-robin index. The model scales naturally to more producers, consumers, or requests by adjusting the domain definitions. The three agent procedures are given in Figure~\ref{fig:lb-bach}.
\end{sloppypar}

        The case study uses a service domain \texttt{Service = \{pzero, loadBalancer, czero, cone\}}, a bounded integer type \texttt{MyInteger = \{zero, one, two\}}, and arithmetic helpers: \texttt{toInt}, \texttt{successor}, and \texttt{mod2}, whose equations are omitted for brevity. To track the routing of requests, messages follow the structure \texttt{message(source, destination, context, id)}, where \texttt{context} preserves the original sender or final recipient throughout forwarding. The load balancer alternates between consumers based on the parity of \texttt{roundRobinIndex} (even $\to$ \texttt{czero}, odd $\to$ \texttt{cone}).

        \begin{remark}  
        \texttt{successor} is defined cyclically with \texttt{successor(two) = zero}, ensuring that the load balancer's recursive call $\texttt{LoadBalancer(successor(roundRobinIndex))}$ remains well-defined when \texttt{roundRobinIndex} reaches \texttt{two}. The wraparound has no observable effect: after this final increment, the producer has already stopped emitting requests (its guard $(\texttt{reqId} \neq \texttt{maxReqId})$ halts emission once \texttt{maxReqId} is reached), so the load balancer never processes another request and the wrapped index is never used. Consistency between the bound of \texttt{MyInteger} and \texttt{maxReqId} is therefore an important modelling invariant when scaling to larger configurations.
        \end{remark}

    \subsection{Generated mCRL2 Model}

        The Bach specification is automatically translated into mCRL2 by our tool following the framework described in Section~\ref{sec:bachTomCRL2}. The generated code instantiates the Blackboard pattern of Figure~\ref{fig:blackboard} with the Load Balancer procedures, wrapped in the standard \texttt{allow/comm} synchronization block. The full mCRL2 specification is available in the accompanying repository~\cite{repo}.

    \subsection{Verification of the Load Balancer Properties}

        Our automated translation tool fully automates the translation of properties expressible in Anemone — using \texttt{Reach}, \texttt{Next}, \texttt{Until} together with simple counting predicates of the form \texttt{\#a = c} — by generating both the corresponding $\mu$-calculus formula and the Blackboard instrumentation.

        The properties of this section, however, exceed Anemone's expressiveness. Their action-based formulations were written manually in $\mu$-calculus. Their content-aware formulations were either adapted from automatically generated templates and manually edited (No Duplication), or written from scratch while reusing the Blackboard instrumentation patterns produced by the tool for simpler predicates (Absence of Loss) — for example, the single-counter pattern \texttt{\#a = c} provides the structural template from which composed predicates such as \texttt{\#a + \#b > 1} are derived. The tool's contribution thus extends beyond fully automatic translation: it also provides reusable patterns that reduce the manual effort required for richer properties.

        \paragraph{No Duplication of Requests.}
            This property ensures that a given request identifier is never processed twice by the system.
            
            \noindent\textit{Action-based formulation.}
            The property is expressed by asserting that no execution path contains two consecutive occurrences of the same request identifier:
                
                \begin{lstlisting}[mathescape=true,basicstyle=\footnotesize\ttfamily]
forall reqId: MyInteger.
    !<true* . Told(request(pzero,loadBalancer,reqId)) .
        true* . Told(request(pzero,loadBalancer,reqId))> true
                \end{lstlisting}

            \noindent\textit{Content-aware formulation.}
             To express the property using the Blackboard's content, a natural first attempt is to extend the Blackboard with an observable action triggered when two occurrences of the same request are simultaneously present:
                \begin{lstlisting}[basicstyle=\footnotesize\ttfamily]
... sum id: MyInteger. 
        (count(request(pzero,loadBalancer,id),myBB) > 1) 
            -> BB_dup_req . Blackboard(myBB)
                \end{lstlisting}
            However, this approach has a fundamental limitation. The Blackboard reflects the \emph{current state} of the shared space at a given moment $t$, not its full history. A request could be emitted a first time, consumed by the LoadBalancer via \texttt{get}, and then emitted again — constituting a genuine duplication — yet the count in the Blackboard would never exceed 1 at any single instant. The property would therefore incorrectly evaluate to \texttt{true}. 
            A more robust approach exploits the fact that, in this model, every request eventually leads to a success si-term (since no interference is assumed), and crucially, these terms are never consumed. Duplication can therefore be reliably detected at the level of success si-term, which accumulate in the shared space:
                    \begin{lstlisting}[basicstyle=\footnotesize\ttfamily]
... sum id: MyInteger.
        (count(success(czero,pzero,id),myBB) +
        count(success(cone,pzero,id),myBB) > 1)
            -> BB_dup_req . Blackboard(myBB)
                    \end{lstlisting}

            This yields the following formula, valid under the assumption that each request leads to exactly one success — which holds in this interference-free model but may not generalize to models with interference:

                    \[
                    ! \langle \mathit{true}^*\ .\ \texttt{BB\_dup\_req} \rangle\, \mathit{true}
                    \]

    \paragraph{Cyclic Distribution of Requests.}

        This property ensures that requests are distributed in a strict round-robin fashion between the two consumers. It is naturally expressed using the action-based approach, as it relies on the \emph{order} in which forwarding actions occur rather than on the contents of the shared space at a given moment. The Blackboard emits observable actions whenever a condition holds on its current contents, independently of when the corresponding data was inserted. Verifying alternation therefore requires tracking the \emph{sequence} of observable forwarding actions, which is precisely what the action-based $\mu$-calculus excels at.

\begin{figure}[!h]
            \begin{minipage}{\linewidth}
            \begin{lstlisting}[mathescape=true, basicstyle=\footnotesize\ttfamily]
nu X(last: Service = cone) .(
    [!(exists id: MyInteger.
        Told(request(loadBalancer, czero, id)) ||
        Told(request(loadBalancer, cone, id))
    )] X(last)
    &&
    (forall id: MyInteger.
        [Told(request(loadBalancer, czero, id))]
        (val(last == cone) && X(czero))
    )
    &&
    (forall id: MyInteger.
        [Told(request(loadBalancer, cone, id))]
        (val(last == czero) && X(cone))
    ))
            \end{lstlisting}
            \end{minipage}
\caption{$\mu$-calculus formula for the Cyclic Distribution property.\label{figure-nu}}
\end{figure}

        As shown in Figure~\ref{figure-nu}, the formula uses a parameterised fixpoint $nu X(\texttt{last})$ whose parameter records the consumer targeted by the last forwarding action, with initial value \texttt{cone} so that the very first forwarding is constrained to go to \texttt{czero}. The formula combines three conjuncts. The first states that any action which is \emph{not} a forwarding to a consumer preserves the invariant without changing \texttt{last}. The second states that any forwarding to \texttt{czero} requires the previous forwarding to have been towards \texttt{cone}, and updates \texttt{last} to \texttt{czero}. The third is symmetric: any forwarding to \texttt{cone} requires the previous forwarding to have been towards \texttt{czero}, and updates \texttt{last} accordingly. Together, these three clauses enforce strict alternation along every execution.

        This example illustrates a general principle: content-aware verification via Blackboard predicates is most effective for \emph{state-based} properties (e.g., presence or count of data), while action-based reasoning is more suited to \emph{ordering} and \emph{sequencing} properties. The two approaches are therefore complementary.

        \paragraph{Absence of Loss.}

            This property ensures that no request sent by the producer is blocked before a successful response
    
            \noindent\textit{Action-based formulation.} The action-based formulation is given in Figure~\ref{fig-nu-2}. The outer $nu X$ enforces the property as a global invariant. Upon each request, the inner $\nu Y$ checks that the system never reaches a deadlock before a success is observed, the modality $\langle \mathit{true} \rangle \mathit{true}$ guaranteeing that some action always remains enabled:

\begin{figure}[!h]
            \begin{minipage}{\linewidth}
            \begin{lstlisting}[mathescape=true,basicstyle=\footnotesize\ttfamily]
forall reqId: MyInteger.
    nu X. (
        [true] X
        &&
        [Told(request(pzero, loadBalancer, reqId))]
            nu Y. (
                [!Told(success(czero, pzero, reqId))
                && !Told(success(cone,  pzero, reqId))] Y
                && <true> true
            )
    )
            \end{lstlisting}
            \end{minipage}
            \caption{Action-based formulation of the Absence of Loss property.\label{fig-nu-2}}
            \end{figure}

        \noindent\textit{Content-aware formulation.} \label{AL_BB} We extend the Blackboard with an observable action emitted whenever exactly one success message for a given request is present:
        \begin{lstlisting}[mathescape=true,basicstyle=\footnotesize\ttfamily]
... sum id: MyInteger.
        (count(success(czero, pzero, id), myBB) +
         count(success(cone,  pzero, id), myBB) == 1)
            -> BB_success(id) . Blackboard(myBB)
        \end{lstlisting}

        The condition \texttt{== 1} simultaneously encodes two guarantees: the request has been processed (\texttt{count} $\geq 1$) and has not been duplicated (\texttt{count} $\leq 1$). Building on this instrumentation, the content-aware formulation replaces the pair of \texttt{Told(success(...))} observations of the action-based version by a single Blackboard predicate \texttt{BB\_success(reqId)}, which folds the two consumer cases into one signal:

            \begin{minipage}{\linewidth}
            \begin{lstlisting}[mathescape=true, basicstyle=\footnotesize\ttfamily]
forall reqId: MyInteger.
    nu X. (
        [true] X
        &&
        [Told(request(pzero, loadBalancer, reqId))]
        nu Y. ( [!BB_success(reqId)] Y && <true>true )
    )
            \end{lstlisting}
            \end{minipage}

    \subsection{Detecting Violations}\label{sec:violation}
        Beyond confirming expected properties, the framework also helps a developer pinpoint genuine bugs in the model. To illustrate this, suppose the round-robin logic of the load balancer is accidentally broken — for example, by hard-coding every forwarding to \texttt{czero}. Verifying the cyclic distribution property on this faulty variant returns \texttt{false}, together with a counter-example trace whose forwarding actions are:
\begin{lstlisting}[mathescape=true, basicstyle=\footnotesize\ttfamily]
... Got(request(pzero, loadBalancer, zero))
Told(request(loadBalancer, czero, pzero, zero))
... Got(request(pzero, loadBalancer, one))
Told(request(loadBalancer, czero, pzero, one))   <- violates alternation
\end{lstlisting}
        The trace is expressed entirely in the Bach vocabulary preserved by our translation (Section~\ref{sec:bachTomCRL2}): the actions \texttt{Told}, \texttt{Got}, and the si-terms appear exactly as they were written in the original Bach program. Filtering on forwarding actions reveals the violation directly: two consecutive routings to \texttt{czero} indicate the broken alternation, without any need to inspect the $\mu$-calculus formula. 
        The cognitive cost of $\mu$-calculus is thus concentrated at specification time, when properties are first written; once the property is in place, diagnosing a violation typically reduces to reading a sequence of familiar Bach actions, which is comparable to inspecting a regular program trace and requires little to no $\mu$-calculus expertise.

\section{Conclusion and Discussion}\label{sec:conclusion}

    This paper presents a framework for the automated verification of Bach coordination programs using mCRL2. Starting from a formal translation of Bach into mCRL2, we establish a trace correspondence theorem guaranteeing the observable behavior of a Bach agent is faithfully preserved in its mCRL2 encoding. Building on this foundation, we introduce a content-aware verification approach that extends the standard action-based $\mu$-calculus with predicates directly expressed over the shared space contents, by instrumenting the Blackboard process with dedicated observable actions.

    Both contributions are supported by an accompanying tool. The translation interface lets the user enter a Bach program in a text area, click \emph{Translate}, and obtain the corresponding mCRL2 specification, with a side panel summarising Bach syntax and coordination primitives for reference. The proof translation interface accepts a Bach temporal property as input and automatically generates the four mCRL2 artifacts required for content-aware verification: the action declaration (\texttt{act}), the Blackboard branch (\texttt{proc Blackboard}), the \texttt{allow} update, and the $\mu$-calculus formula. Both interfaces are available in the accompanying repository~\cite{repo}.

    The Load Balancer case study illustrates the complementary aspects of the two verification styles. Action-based formulations naturally capture \emph{ordering} and \emph{history-based} properties — such as cyclic distribution, where alternation between consumers must be tracked over time. Content-aware formulations excel at \emph{state-based} properties over persistent data — such as absence of loss, where accumulated success messages provide a reliable witness. The no-duplication property further illustrates that naive content-aware encodings may fail when si-terms are consumed, but can be recovered by shifting observation to persistent effects. A useful guideline emerges: content-aware verification is most effective for state-based properties over persistent data, while history-based properties are more naturally captured through action-based reasoning.

    Beyond verifying expected properties, the framework also supports the practical task of \emph{detecting errors in the model}. When a property fails, the mCRL2 model checker returns a counter-example trace expressed entirely in the Bach vocabulary preserved by our encoding — the original primitives (\texttt{Told}, \texttt{Got}, \texttt{Asked}, \texttt{Nasked}) and the original si-terms. As shown on the broken round-robin variant (Section~\ref{sec:violation}), a Bach developer can often localise the fault directly from the trace, without inspecting the $\mu$-calculus formula. The cognitive cost of $\mu$-calculus is thus concentrated at \emph{specification time}, when temporal properties are first written, while \emph{diagnosis time} remains close to reading a familiar Bach execution.

    However, the content-aware approach has inherent limitations. Instrumenting the Blackboard introduces additional observable actions and increases the state space, which may impact verification performance on larger models. Furthermore, some content-aware encodings rely on  model-specific assumptions — such as the absence of interference — that may not hold in more general settings. Such instrumentation should therefore be applied with care.

    \paragraph{Future Work.}
    Several directions remain open. On the theoretical side, the trace correspondence established in this paper can be strengthened by investigating whether a weak bisimulation holds between a Bach agent and its mCRL2 encoding, which would provide a finer behavioral characterization of the translation.
    On the practical side, while this paper introduced a verification approach based on a shared-space model, its empirical exploration remains limited, and a more systematic evaluation of its strengths and limitations on richer coordination patterns remains an open challenge. In particular, applying the framework to more complex patterns such as the Circuit Breaker and the Message Broker would be a natural extension of this work.
    A third direction concerns the expressiveness of the temporal language itself: extending Anemone's operators and our automated translation to richer patterns would reduce the manual effort currently required for properties beyond Anemone's expressiveness, as illustrated by some of the Load Balancer properties of Section~\ref{sec:usecase}.
    Finally, although it addresses a coordination language of a different nature, named Reo, which relies on a control-based approach as opposed to the data-based approach used by Bach, the article~\cite{Kokash-et-al} proposes other techniques for translating service compositions into mCRL2. A comparison between these techniques and those developed in the present paper would be worthwhile, particularly in light of the relationships already highlighted in the articles~\cite{JA-LI-DD-Fetschrift,KrauseMLA11}.

\section{Acknowledgment}

The authors thank the University of Namur for its support. They also
thank the Walloon Region for partial support through the Ariac project
(convention 210235) and the CyberExcellence project (convention 2110186).

\bibliographystyle{eptcs}
\bibliography{bibi}

\begin{thebibliography}{10}
\providecommand{\bibitemdeclare}[2]{}
\providecommand{\surnamestart}{}
\providecommand{\surnameend}{}
\providecommand{\urlprefix}{Available at }
\providecommand{\url}[1]{\texttt{#1}}
\providecommand{\href}[2]{\texttt{#2}}
\providecommand{\urlalt}[2]{\href{#1}{#2}}
\providecommand{\doi}[1]{doi:\urlalt{https://doi.org/#1}{#1}}
\providecommand{\eprint}[1]{arXiv:\urlalt{https://arxiv.org/abs/#1}{#1}}
\providecommand{\bibinfo}[2]{#2}

\bibitemdeclare{inproceedings}{BaJa-ICE23}
\bibitem{BaJa-ICE23}
\bibinfo{author}{M.~\surnamestart Barkallah\surnameend} \&
  \bibinfo{author}{J.-M. \surnamestart Jacquet\surnameend}
  (\bibinfo{year}{2023}): \emph{\bibinfo{title}{{On the Introduction of Guarded
  Lists in Bach: Expressiveness, Correctness, and Efficiency Issues}}}.
\newblock In \bibinfo{editor}{C.~\surnamestart Aubert\surnameend},
  \bibinfo{editor}{C.~Di \surnamestart Giusto\surnameend},
  \bibinfo{editor}{S.~\surnamestart Fowler\surnameend} \&
  \bibinfo{editor}{L.~\surnamestart Safina\surnameend}, editors: {\slshape
  \bibinfo{booktitle}{Proceedings 16th Interaction and Concurrency Experience
  (ICE) 2023}}, {\slshape \bibinfo{series}{{EPTCS}}} \bibinfo{volume}{383}, pp.
  \bibinfo{pages}{55--72}, \doi{10.4204/EPTCS.383.4}.

\bibitemdeclare{article}{barkallahsocio}
\bibitem{barkallahsocio}
\bibinfo{author}{Manel \surnamestart Barkallah\surnameend}:
  \emph{\bibinfo{title}{{On Reasoning about Socio-Technical Systems: the
  Multi-Bach Coordination Model and its Workbench Anemone}}}.

\bibitemdeclare{article}{bergstra1985algebra}
\bibitem{bergstra1985algebra}
\bibinfo{author}{Jan~A. \surnamestart Bergstra\surnameend} \&
  \bibinfo{author}{Jan~Willem \surnamestart Klop\surnameend}
  (\bibinfo{year}{1985}): \emph{\bibinfo{title}{Algebra of communicating
  processes with abstraction}}.
\newblock {\slshape \bibinfo{journal}{Theoretical computer science}}
  \bibinfo{volume}{37}, pp. \bibinfo{pages}{77--121},
  \doi{10.1016/0304-3975(85)90088-X}.

\bibitemdeclare{article}{carriero1989linda}
\bibitem{carriero1989linda}
\bibinfo{author}{Nicholas \surnamestart Carriero\surnameend} \&
  \bibinfo{author}{David \surnamestart Gelernter\surnameend}
  (\bibinfo{year}{1989}): \emph{\bibinfo{title}{Linda in context}}.
\newblock {\slshape \bibinfo{journal}{Communications of the ACM}}
  \bibinfo{volume}{32}(\bibinfo{number}{4}), pp. \bibinfo{pages}{444--458}.

\bibitemdeclare{inproceedings}{DJL18}
\bibitem{DJL18}
\bibinfo{author}{D.~\surnamestart Darquennes\surnameend},
  \bibinfo{author}{J.-M. \surnamestart Jacquet\surnameend} \&
  \bibinfo{author}{I.~\surnamestart Linden\surnameend} (\bibinfo{year}{2018}):
  \emph{\bibinfo{title}{{On Multiplicities in Tuple-Based Coordination
  Languages: The Bach Family of Languages and Its Expressiveness Study}}}.
\newblock In \bibinfo{editor}{G.~Di~Marzo \surnamestart Serugendo\surnameend}
  \& \bibinfo{editor}{M.~\surnamestart Loreti\surnameend}, editors: {\slshape
  \bibinfo{booktitle}{Proceedings of the 20th International Conference on
  Coordination Models and Languages}}, {\slshape \bibinfo{series}{Lecture Notes
  in Computer Science}} \bibinfo{volume}{10852}, \bibinfo{publisher}{Springer},
  pp. \bibinfo{pages}{81--109}, \doi{10.1007/978-3-319-92408-3\_4}.

\bibitemdeclare{article}{de2002klaim}
\bibitem{de2002klaim}
\bibinfo{author}{Rocco \surnamestart De~Nicola\surnameend},
  \bibinfo{author}{Gian-Luigi \surnamestart Ferrari\surnameend} \&
  \bibinfo{author}{Rosario \surnamestart Pugliese\surnameend}
  (\bibinfo{year}{2002}): \emph{\bibinfo{title}{KLAIM: A kernel language for
  agents interaction and mobility}}.
\newblock {\slshape \bibinfo{journal}{IEEE Transactions on software
  engineering}} \bibinfo{volume}{24}(\bibinfo{number}{5}), pp.
  \bibinfo{pages}{315--330}, \doi{10.1109/32.685256}.

\bibitemdeclare{inproceedings}{de2005formal}
\bibitem{de2005formal}
\bibinfo{author}{Rocco \surnamestart De~Nicola\surnameend},
  \bibinfo{author}{Diego \surnamestart Latella\surnameend} \&
  \bibinfo{author}{Mieke \surnamestart Massink\surnameend}
  (\bibinfo{year}{2005}): \emph{\bibinfo{title}{Formal modeling and
  quantitative analysis of KLAIM-based mobile systems}}.
\newblock In: {\slshape \bibinfo{booktitle}{Proceedings of the 2005 ACM
  symposium on Applied computing}}, pp. \bibinfo{pages}{428--435},
  \doi{10.1145/1066677.1066777}.

\bibitemdeclare{article}{gelernter1992coordination}
\bibitem{gelernter1992coordination}
\bibinfo{author}{David \surnamestart Gelernter\surnameend} \&
  \bibinfo{author}{Nicholas \surnamestart Carriero\surnameend}
  (\bibinfo{year}{1992}): \emph{\bibinfo{title}{Coordination languages and
  their significance}}.
\newblock {\slshape \bibinfo{journal}{Communications of the ACM}}
  \bibinfo{volume}{35}(\bibinfo{number}{2}), p.~\bibinfo{pages}{96},
  \doi{10.1145/129630.376083}.

\bibitemdeclare{article}{gelernter1985generative}
\bibitem{gelernter1985generative}
\bibinfo{author}{\surnamestart {Gelernter, David}\surnameend}
  (\bibinfo{year}{1985}): \emph{\bibinfo{title}{{Generative communication in
  Linda}}}.
\newblock {\slshape \bibinfo{journal}{{ACM Transactions on Programming
  Languages and Systems (TOPLAS)}}} \bibinfo{volume}{7}(\bibinfo{number}{1}),
  pp. \bibinfo{pages}{80--112}, \doi{10.1145/2363.2433}.

\bibitemdeclare{book}{groote2014modeling}
\bibitem{groote2014modeling}
\bibinfo{author}{Jan~Friso \surnamestart Groote\surnameend} \&
  \bibinfo{author}{Mohammad~Reza \surnamestart Mousavi\surnameend}
  (\bibinfo{year}{2014}): \emph{\bibinfo{title}{Modeling and analysis of
  communicating systems}}.
\newblock \bibinfo{publisher}{MIT press}, \doi{10.7551/mitpress/9946.003.0010}.

\bibitemdeclare{article}{groote2009analysis}
\bibitem{groote2009analysis}
\bibinfo{author}{\surnamestart {Groote, Jan Friso and Mathijssen, Aad and
  Reniers, Michel A and Usenko, Yaroslav S and van Weerdenburg,
  Muck}\surnameend} (\bibinfo{year}{2009}): \emph{\bibinfo{title}{{Analysis of
  distributed systems with mCRL2}}}.
\newblock {\slshape \bibinfo{journal}{{Process Algebra for Parallel and
  Distributed Processing}}}.

\bibitemdeclare{article}{BaJa-Anemone-21}
\bibitem{BaJa-Anemone-21}
\bibinfo{author}{J.-M. \surnamestart Jacquet\surnameend} \&
  \bibinfo{author}{M.~\surnamestart Barkallah\surnameend}
  (\bibinfo{year}{2021}): \emph{\bibinfo{title}{{Anemone: {A} workbench for the
  Multi-Bach coordination language}}}.
\newblock {\slshape \bibinfo{journal}{Science of Computer Programming}}
  \bibinfo{volume}{202}, p. \bibinfo{pages}{102579},
  \doi{10.1016/J.SCICO.2020.102579}.

\bibitemdeclare{inproceedings}{JL07}
\bibitem{JL07}
\bibinfo{author}{J.-M. \surnamestart Jacquet\surnameend} \&
  \bibinfo{author}{I.~\surnamestart Linden\surnameend} (\bibinfo{year}{2007}):
  \emph{\bibinfo{title}{Coordinating {C}ontext-aware {A}pplications in {M}obile
  {A}d-hoc {N}etworks}}.
\newblock In \bibinfo{editor}{T.~\surnamestart Braun\surnameend},
  \bibinfo{editor}{D.~\surnamestart Konstantas\surnameend},
  \bibinfo{editor}{S.~\surnamestart Mascolo\surnameend} \&
  \bibinfo{editor}{M.~\surnamestart Wulff\surnameend}, editors: {\slshape
  \bibinfo{booktitle}{Proceedings of the first ERCIM workshop on eMobility}},
  \bibinfo{publisher}{The University of Bern}, pp. \bibinfo{pages}{107--118}.

\bibitemdeclare{inproceedings}{JA-LI-DD-Fetschrift}
\bibitem{JA-LI-DD-Fetschrift}
\bibinfo{author}{J.{-}M. \surnamestart Jacquet\surnameend},
  \bibinfo{author}{I.~\surnamestart Linden\surnameend} \&
  \bibinfo{author}{D.~\surnamestart Darquennes\surnameend}
  (\bibinfo{year}{2018}): \emph{\bibinfo{title}{{On the Relation Between
  Control-Based and Data-Based Coordination Languages}}}.
\newblock In \bibinfo{editor}{F.S. \surnamestart de~Boer\surnameend},
  \bibinfo{editor}{M.M. \surnamestart Bonsangue\surnameend} \&
  \bibinfo{editor}{J.~\surnamestart Rutten\surnameend}, editors: {\slshape
  \bibinfo{booktitle}{{It's All About Coordination - Essays to Celebrate the
  Lifelong Scientific Achievements of Farhad Arbab}}}, \bibinfo{series}{Lecture
  Notes in Computer Science}, \bibinfo{publisher}{Springer}, pp.
  \bibinfo{pages}{86--106}, \doi{10.1007/978-3-319-90089-6_7}.

\bibitemdeclare{inproceedings}{Kokash-et-al}
\bibitem{Kokash-et-al}
\bibinfo{author}{N.~\surnamestart Kokash\surnameend},
  \bibinfo{author}{C.~\surnamestart Krause\surnameend} \& \bibinfo{author}{E.P.
  \surnamestart de~Vink\surnameend} (\bibinfo{year}{2010}):
  \emph{\bibinfo{title}{{Data-aware Design and Verification of Service
  Compositions with Reo and mCRL2}}}.
\newblock In \bibinfo{editor}{S.Y. \surnamestart Shin\surnameend},
  \bibinfo{editor}{S.~\surnamestart Ossowski\surnameend},
  \bibinfo{editor}{M.~\surnamestart Schumacher\surnameend},
  \bibinfo{editor}{M.J. \surnamestart Palakal\surnameend} \&
  \bibinfo{editor}{C.{-}C. \surnamestart Hung\surnameend}, editors: {\slshape
  \bibinfo{booktitle}{Proceedings of the 2010 {ACM} Symposium on Applied
  Computing (SAC)}}, \bibinfo{publisher}{{ACM}}, pp.
  \bibinfo{pages}{2406--2413}, \doi{10.1145/1774088.1774590}.

\bibitemdeclare{article}{KrauseMLA11}
\bibitem{KrauseMLA11}
\bibinfo{author}{C.~\surnamestart Krause\surnameend},
  \bibinfo{author}{Z.~\surnamestart Maraikar\surnameend},
  \bibinfo{author}{A.~\surnamestart Lazovik\surnameend} \&
  \bibinfo{author}{F.~\surnamestart Arbab\surnameend} (\bibinfo{year}{2011}):
  \emph{\bibinfo{title}{{Modeling Dynamic Reconfigurations in Reo using
  High-level Replacement Systems}}}.
\newblock {\slshape \bibinfo{journal}{Science of Computer Programming}}
  \bibinfo{volume}{76}(\bibinfo{number}{1}), pp. \bibinfo{pages}{23--36},
  \doi{10.1016/j.scico.2009.10.006}.

\bibitemdeclare{}{reutersAWS2025}
\bibitem{reutersAWS2025}
\bibinfo{author}{\surnamestart {Reuters}\surnameend} (\bibinfo{year}{2025}):
  \emph{\bibinfo{title}{{Amazon says AWS cloud service back to normal after
  outage disrupts businesses worldwide}}}.
\newblock
  \urlprefix\url{https://www.reuters.com/business/retail-consumer/amazons-cloud-unit-reports-outage-several-websites-down-2025-10-20/}.
\newblock \bibinfo{note}{Accessed: 27 Dec. 2025}.

\bibitemdeclare{misc}{repo}
\bibitem{repo}
\bibinfo{author}{C.~\surnamestart Reuther\surnameend} \& \bibinfo{author}{J.-M.
  \surnamestart Jacquet\surnameend} (\bibinfo{year}{2026}):
  \emph{\bibinfo{title}{Bach-to-mCRL2 Translation Tool and Case Studies}}.
\newblock
  \bibinfo{howpublished}{\url{https://github.com/UNamurCSFaculty/anemone-bach-to-mcrl2}}.

\end{thebibliography}
\end{document}